\documentclass[hidelinks, 12pt]{article}
\usepackage[utf8]{inputenc}
\usepackage{mathtools}   
\usepackage[T1]{fontenc}
\usepackage{fullpage}
\usepackage{graphicx}
\usepackage{amsfonts}
\usepackage{amsmath}
\usepackage{amsthm}
\usepackage{algorithm,algorithmic}
\usepackage{hyperref}
\newtheorem*{prop*}{Proposition}
\title{Replicating Portfolios: Constructing Permissionless Derivatives}
\author{ 
Estelle Sterrett\\
{\small \texttt{estelle@primitive.xyz}}
\and
Waylon Jepsen\\
{\small \texttt{jepsen@primitive.xyz}}
\and
Evan Kim\\
{\small \texttt{evan@primitive.xyz}}
}
\date{May 2022}

\begin{document}

\maketitle
\begin{abstract}
    The current design space of derivatives in Decentralized Finance (DeFi) relies heavily on oracle systems. Replicating market makers (RMMs) provide a mechanism for converting specific payoff functions to an associated Constant Function Market Makers (CFMMs). We leverage RMMs to replicate the approximate payoff of a Black-Scholes covered call option. RMM-01 allows access to Black-Scholes pricing on-chain without oracles and with the robustness of CFMMs. We provide frameworks for derivative instruments and structured products achievable on-chain structured around RMM-01. We construct long and binary options and briefly discuss perpetual covered call strategies commonly referred to as "theta vaults". Moreover, we introduce a procedure to eliminate liquidation risk in lending markets. The results suggest that CFMMs are essential for structured product design with minimized trust dependencies.
\end{abstract}
\section*{Introduction}

\paragraph{Order Books.}
Traditional exchanges used for trading stocks, commodities, and even digital assets follow a continuous-limit order book design. This market design consists of a list of open buy and sells orders. Buyers and sellers place order limits that specify the maximum or minimum price at which they are willing to buy or sell an asset. A centralized counterparty then matches the buy and sell orders automatically and profits on \emph{making the market}. Contrasting the centralized exchange, decentralized exchanges (DEXs) manage trades on-chain with a smart contract. This design provides a transparent and censorship-resistant marketplace \cite{daian2019flash, uniswap}.
\paragraph{CFMMs.}
Constant Function Market Makers (CFMMs) dominate the design space of DEXs~\cite{uniswap,angeris2020improved,angeris2022optimal}. CFMMs account for tens of billions of liquidity and trillions in trading volume~\cite{primitive}. In a CFMM, users provide liquidity to a vector of asset reserves $R\in\mathbb{R}^n_{+}$ where $n \geq 2$ in exchange for a liquidity provider token (LPT) representing their share of the total reserves. The LPTs serve as an immutable claim to the proportion of the liquidity pool a user owns. A trading function $\varphi: \mathbb{R}_{+}^n \rightarrow \mathbb{R}$ governs the CFMM where a trade impacts the reserves $\varphi(R') \geq \varphi(R)$ so that the value of the function after a trade is greater than or equal to the value of the function before~\cite{angeris2021replicating}. For example, the trading function of \emph{Uniswap V2}  is represented by the the function $\varphi(R_x, R_y) = R_{x}R_y$~\cite{Zinsmeister_v2-core_2019}.

\paragraph{Oracles.}
Until RMM-01, the methodology for providing financial derivative products involved the oracle dependency. Oracle dependencies introduce a centralization attack vector that significantly increases risk~\cite{app11167572}. These price oracle attacks have been orchestrated in production environments~\cite{hackernoon} and can be detrimental to protocol users. Recent work~\cite{angeris2020improved} leverages arbitrage for price alignment~\cite{cox1979option} unveiling how to mitigate oracle dependencies. Oracle dependencies dissuade protocol users from allocating more liquidity than the oracle source. 
\paragraph{RMMs.}
Replicating Market Makers (RMMs) allow protocol designers to directly provide a 1-homogeneous, concave, non-negative, non-decreasing payoff function to the liquidity providers via a replicating portfolio~\cite{angeris2021replicating}. Financial derivatives have been prevalent throughout history in traditional finance, allowing individuals to take leveraged or hedged directional views of different markets. As of March 2022, the derivative market on centralized exchanges represents 62.8\% of trading volume~\cite{cryptocompare} expressing an apparent demand. We can construct a collection of financial products familiar to traditional financial market participants by using the liquidity tokens from RMM-01 and the underlying risky and stable assets.

\paragraph{RMM-01.}
RMM-01~\cite{primitive} is a time-sensitive trading function
\begin{equation}\label{tradingfunc}
    \varphi(R_x,R_y) = R_y-K\Phi(\Phi^{-1}(1-R_x)-\sigma\sqrt{\tau})
\end{equation}
where $K$, $\sigma$, and $\tau$ are the options parameters, and $R_x$ and $R_y$ are the risky and stable asset reserves. Without a fee regime on trades, the LPTs fail to replicate a Black-Scholes covered call due to the external funding required to capture theta decay. The unique optimal fee regime in~\cite{Experience_rmms-py_2022} reduces the difference between the payoff of the LPT and that of a Black-Scholes covered call option to a near negligible margin. $K$ is the strike price in traditional option literature~\cite{black2019pricing}, while $k = \varphi(R_x,R_y)$ is the invariant as specified commonly in CFMM literature~\cite{uniswap}. Let $V_{cc}$ denote the value of a Black-Scholes covered call. In RMM-01, the invariant $k$ can take either a negative or positive value representing $V_{LPT}-V_{cc}$. At $k=0$, the RMM-01 pool is replicating a covered call exactly.
\paragraph{Liquidity Provider Tokens.}
RMM-01 LPTs are tokenized under the ERC-1155 standard~\cite{tokenStandards}. The composable nature of ERC-1155 tokens and the oracle-free Black-Scholes covered call from RMM-01 allow for the construction of every other Black-Scholes options instrument barring liquidity constraints. RMM-01 native derivatives are subject to a terminal error defined in more detail in~\cite{Experience_rmms-py_2022}. Some of these RMM-01 based mechanisms require an incentive for LPT holders to create available liquidity and consequently may be slightly overpriced with respect to Black-Scholes pricing~\cite{black2019pricing}. Overpricing is proportional to the incentives for liquidity providers. Within the same pool, additional layers of composability create liquidity fragmentation, rendering the implementation of short options unrealistic.
\paragraph{Outline.}
In~\S\ref{sec:VannillaOptions} we show the construction of vanilla options. In~\S\ref{sec:Binary} we move into the exotic Binary European Options. In~\S\ref{sec:thetavaults} we introduce a construction of the transparent structured product~\cite{Kallio2022-al} theta vaults introduced by the \emph{Opyn} protocol~\cite{ConvexityProtocol}. Next, in~\S\ref{sec:liq} we introduce a mechanism for a liquidation-free lending market that leverages RMM-01 constructions to ensure a hedged health factor. The last section~\S\ref{sec:combo} combines some of the potentially existing constructions to replicate the payoff of a straddle and a futures contract payoff. This approach to derivative products is more secure than oracle-based systems under the condition of sufficient liquidity in the underlying RMM-01. Each mechanism that involves a lending flow assumes no swap friction. To account for swap friction, over-collateralize proportionally to the slippage of the underlying market. 

\section{Vanilla Options}\label{sec:VannillaOptions}
\paragraph{}
We can achieve long options through a borrowing and shorting flow on the LPTs. Applying a theoretical examination to long options, suppose there is an RMM-01 pool with the strike, expiry, and implied volatility $K, T,\sigma$, respectively. The payoff of the LPTs denominated in the cash asset follows
\begin{equation}\label{LPT Payoff}
    V_{LPT}(S(t)) = S(t)(\Phi(-d_1))+K\Phi(d_2)+k,
\end{equation}
where $S(t)$ is the reported price of the underlying pool at time $t<T$,
\begin{equation}
    d'_1=\frac{\ln{(\frac{S(t)}{K})}+(\sigma^2/2)\tau}{\sigma\sqrt{\tau}} \\,
\end{equation}
and $d'_2=d'_1-\sigma\sqrt{\tau}$ with $\tau=T-t$ and $\Phi$ is the standard normal cumulative distribution function. The payoff of a long call follows $V_{call} = \Phi(d_{1}(t)) - \frac{K}{S(t)}\Phi(d_{2}(t))$ and the payoff of a long put follows $V_{put}=K\Phi(-d_2(t))-S(t)\Phi(-d_1(t))$~\cite{black2019pricing}.

To obtain either of the long options, we need to short the LPT for either the cash asset or risky asset, the first leading to a long put and the latter leading to a long call. To achieve this without oracle dependencies~\cite{app11167572} or a liquidation framework~\cite{qin2021empirical}, we need to maintain enough collateral to cover the maximum possible debt of $K+k$ units of stable assets. Since the invariant $k$ is unknown prior to expiry, the maximum debt is variable~\cite{primitive}. If $k$ is strongly positive, the value of the LPTs is significantly greater than that of a Black-Scholes covered call~\cite{black2019pricing}, which could lead to a lack of available collateral liquidity to fund the repayment. 

We suggest a repayment amount based on the covered call value, simply neglecting the invariant $k$. By forcing this constraint, the borrower's debt is capped at $K$ worth of LPTs while still allowing the lender to receive the payoff of the underlying derivative. This constraint guarantees the lender that they receive a covered call payoff by introducing a buy-side to the system. The assurance of this repayment depends on implementation details. To incentivize this lending further, we suggest implementing an interest rate for borrowing. Additionally, an \emph{Ante Finance} Ante Test~\cite{AnteFinance} on the value of the invariant $k$ ensuring $k\leq0$ may be beneficial, allowing the lender to hedge the loss of additional swap fees by challenging the test on the associated trust market. 
\subsection*{Long Call}
\begin{prop*}
    Borrowing $1$ LPT and selling the underlying assets for the risky asset results in $1$ Black-Scholes call option.
\end{prop*}
\begin{proof}
   Let $t_0$ be the time of borrow and $t_0<t_f\leq T$ be the time of repayment. In this case, we are borrowing the LPT and denominating it in the risky asset, which involves breaking the LPT into its underlying assets and selling the stable for the risky. The repayment amount is given by $R$ and the required collateral to cover the maximum debt is given by $C$:
    \begin{align*}
        R&=
            \frac{S(t)\Phi(-d_1(t_f))+K\Phi(d_2(t_f))}{S(t)}\\
        C&=1-\frac{S(t)\Phi(-d_1(t_0))+K\Phi(d_2(t_0))}{S(t)}. \\
    \end{align*}
    The net payoff of the borrower, denoted by $V$, is then
    \begin{align*}
        V&=
            1-\frac{S(t)\Phi(-d_1(t_f))+K\Phi(d_2(t_f))}{S(t)} \\
        V&=
            \Phi(d_1(t_f))-\frac{K}{S(t)}\Phi(d_2(t_f)) && \text{(by symmetry of $\Phi$)}
    \end{align*}
    We can see the borrower is receiving $V_{call}$.
\end{proof}
\label{sec:Lcall}
\subsection*{Long Put}
We now turn to the construction of long puts.
\begin{prop*}
    Borrowing $1$ LPT and selling the underlying assets for the cash asset results in $1$ Black-Scholes put option.
\end{prop*}
\begin{proof}
    Let $t_0$ be the time of borrow and $t_0<t_f\leq T$ be the time of repayment. In this case, we are borrowing the LPT and denominating it in the stable asset, which involves breaking the LPT into its underlying assets and selling the risky for the stable. The repayment amount is given by $R$ and the required collateral to cover the maximum debt is given by $C$:
    \begin{align*}
        R&=
            S(t)\Phi(-d_1(t_f))+K\Phi(d_2(t_f)) \\
        C&=K-S(t)\Phi(-d_1(t_0))-K\Phi(d_2(t_0)) \\
    \end{align*}
    The net payoff of the borrower is denoted by $V$ is thus
    \begin{align*}
        V&=
            K-S(t)\Phi(-d_1(t_f))-K\Phi(d_2(t_f)) \\
        V&=
            K\Phi(-d_2(t_f))-S(t)\Phi(-d_1(t_f)) && \text{(by symmetry of $\Phi$)}
    \end{align*} 
    We can see the borrower is receiving $V_{put}$.
\end{proof}
\section{Binary Options}\label{sec:Binary}
The following most immediate derivative product built on RMM-01 are binary options~\cite{raw2011binary}. The viability of this construction within the EVM depends on the implementation of the redemption mechanism. Using the LPTs, one can sell the rights to either asset in the liquidity position for an upfront premium, resulting in binary options. Specifically, the cash asset in the LPT replicates a cash-or-nothing call option payoff closely, whereas the risky asset replicates an asset-or-nothing put option closely. This mechanism acts seamlessly concerning the performance of the LPTs, as it does not require removing any LPTs from the underlying RMM-01 pool. 

The payoff of a Black-Scholes cash-or-nothing call denominated in the cash asset follows value function
\begin{equation*}\label{CONC}
    V_{conc}=\Phi(d_2).
\end{equation*}
Similarly the payoff of a Black-Scholes asset-or-nothing put denominated in the cash asset follows the value function
\begin{equation*}\label{AONP}
    V_{aonp}=S(t)\Phi(-d_1).
\end{equation*} 

Given an RMM-01 pool with the strike, expiry, and implied volatility $K, T,\sigma$, we know the payoff of this pool follows the value,
\begin{equation*}\label{LPTVal}
V_{LPT}=S(t)(1-\Phi(d_1))+K\Phi(d_2)+k    
\end{equation*}

Suppose a liquidity provider is willing to sell the rights to one of the assets in $1$ LPT. We show that selling the rights to one of the assets simultaneously creates an asset-or-nothing put and a cash-or-nothing call. The buyer receives the position corresponding to the sold asset.
\subsection*{Construction from RMM-01}
\begin{prop*}
    Selling the rights to either of the underlying assets in an LPT leads to the creation of both an asset-or-nothing put and an cash-or-nothing call, to within a range of error defined by the invariant $k$.
\end{prop*}
\begin{proof}
    Let $t_0$ be the time of sale of one of the reserves. Given equation~\eqref{tradingfunc}, it is shown in~\cite{primitive} that the cash reserve follows $R_y =  K \Phi(\Phi^{-1}(1 - R_x) - \sigma \sqrt{\tau}) + k$ and the spot price follows $S(R_x) = Ke^{\Phi^{-1}(1 - R_x)\sigma \sqrt{\tau}} e^{-\frac{1}{2}\sigma^{2}\tau}$. Isolating for $S^{-1}(R_x)$ to determine the payoffs of both reserves yields
    \begin{align*}
        S&=Ke^{\Phi^{-1}(1-R_x)\sigma\sqrt{\tau}}e^{-\frac{1}{2}\sigma^2\tau} \\
        \ln \left(\frac{S}{K}\right)&=\Phi^{-1}(1-R_x)\sigma\sqrt{\tau}-\frac{1}{2}\sigma^2\tau \\ 
        \frac{\ln(\frac{S}{K})}{\sigma\sqrt{\tau}}+\frac{1}{2}\sigma\sqrt{\tau}&=\Phi^{-1}(1-R_x) \\
        \Phi\left(\frac{\ln(\frac{S}{K})}{\sigma\sqrt{\tau}}+\frac{1}{2}\sigma\sqrt{\tau}\right)&=1-R_x \\
        R_x&=\Phi\left(-\frac{\ln(\frac{S}{K})}{\sigma\sqrt{\tau}}-\frac{1}{2}\sigma\sqrt{\tau}\right) &&  \text{(By symmetry of $\Phi$)}\\ 
        S(t)R_x&=S(t)\Phi(-d_1) .\\
    \end{align*}
    The net payoff is
    \begin{equation}
        V_{R_x}=V_{aonp}.
    \end{equation}
    Now for the cash reserve,
    \begin{align*}
        \Phi^{-1}(1-R_x)&=\frac{\ln(\frac{S}{K})}{\sigma\sqrt{\tau}}+\frac{1}{2}\sigma\sqrt{\tau} \\
        R_y&=K\Phi(\frac{\ln(\frac{S}{K})}{\sigma\sqrt{\tau}}-\frac{1}{2}\sigma\sqrt{\tau})+k \\
        R_y&=K\Phi(d_2)+k .\\
    \end{align*}
    The net payoff is
    \begin{equation}
        V_{R_y}=KV_{conc}+k
    \end{equation}
    We've shown that each reserve individually replicates the payoff of either a cash-or-nothing call or an asset-or-nothing put, to within a range of error determined by the underlying invariant $k$.
    
    As a result, selling the rights to one of these reserves $R_i$ both creates a purchase opportunity for one of these binaries, while leaving the seller with the other binary option plus the premium $P=R_i$ at time $t_0$.
\end{proof}
\subsection*{Shorting Binaries}
Creating a token market on the binaries opens a vulnerability related to the American exercise environment. However, a coincidence of wants approach to shorting avoids this issue and the necessity for tokenization.

By shorting a cash-or-nothing call option for the cash asset, we achieve a cash-or-nothing put option, or by shorting the asset-or-nothing put option for the risky asset, we obtain an asset-or-nothing call option~\cite{carr1999static}. The payoff of a Black-Scholes cash-or-nothing put follows the value function,
\begin{equation}\label{conp}
    V_{conp}=\Phi(-d_2)
\end{equation}
Similarly, the payoff of a Black-Scholes asset-or-nothing call denominated in the cash asset follows the value function,
\begin{equation}\label{aonc}
    V_{aonc}=S(t)\Phi(d_1)
\end{equation}
where $S(t)$ is the spot price of the pair at time $t\le T$.

To short the underlying cash-or-nothing call and asset-or-nothing put positions, there must be an LPT available for borrow. There must also be a coincidence of wants between two borrowers looking to short the two positions. The borrower looking to short the cash-or-nothing call must pay a premium of $K-R_y$, and the borrower of the asset-or-nothing put must pay a premium of $1-R_x$. This allows the LPT to break into the underlying assets and distribute the reserves to the associated borrowers.
\subsection*{Remaining Binaries}
We can now prove that payoffs similar to Black-Scholes cash-or-nothing puts and asset-or-nothing calls are achievable through a similar flow to achieving the long options. We will start with the cash-or-nothing put construction.
\begin{prop*}
    Shorting the stable reserve position $R_y$ creates a cash-or-nothing put position.
\end{prop*}
\begin{proof}
    We can define the repayment amount as $R_y$ at time of expiry $T$. We bound the maximum possible debt to $K$ similar to the long options case. Given the borrowed asset's value at time of borrow defined as $KV_{conc}(t_0)+k_0$ with repayment amount $R$, and collateral requirements $C$
    \begin{align*}
        R&=KV_{conc}(T)\\
        C&=K-KV_{conc}(t_0)-k_0 \\
    \end{align*}
    To achieve a cash-or-nothing put, we need to claim the rights of the borrowed stable reserve for the stable asset by removing the LPT. The net position of the borrower is than,
    \begin{align*}
        V&=K-KV_{conc}(T)\\
    V&=KV_{conp}(T) && \text{(by symmetry of $\Phi$)}
    \end{align*}
    The borrower then is long $KV_{conp}$.
\end{proof}
Now we move on to the asset-or-nothing call construction.
\begin{prop*}
    Shorting the risky reserve position $R_x$ creates an asset-or-nothing call position.
\end{prop*}
\begin{proof}
    Let $t_0$ be the time of borrow. The repayment amount can be defined as $R_x$ at the time of expiry $T$. We bound the maximum possible debt at $1$ unit of the risky asset. Given the borrowed asset's value at time of borrow defined as $V_{aonp}(t_0)$ with repayment amount $R$, and the required collateral $C$ 
    \begin{align*}
        R&=\Phi(-d_1(T))\\
        C&=1-\Phi(-d_1(t_0))\\
    \end{align*}
    The net payoff the borrower receives 
    \begin{align*}
        V&=1-\Phi(-d_1(t_f))\\
        V&=\Phi(d_1(t_f)) && \text{(by symmetry of $\Phi$)}
    \end{align*}
    Given the payoff of an asset-or-nothing call denominated in the risky asset is $V_{aonc}=\Phi(d_1)$, we can see the borrower is long $V_{aonc}$.
\end{proof}

\section{Theta Vaults} \label{sec:thetavaults}
\paragraph{}
To start, we turn to a structure product known as theta vaults, as they are the most immediate product one can construct from RMM-01. A theta vault, pioneered by Ribbon Finance~\cite{ribbon}, is a roll-over options strategy of covered call or put selling to earn and compound yield sourced from theta decay consistently. This strategy traditionally involved minting the options through the \emph{Opyn} protocol~\cite{ConvexityProtocol} as \emph{oTokens} and auctioning them off to buyers. This mechanism encounters liquidity issues proportional to the size of the theta vault resulting from the need for a counter-party to take the opposite side of the bet to earn premiums. Similar liquidity problems would arise from the same mechanism in traditional equity markets. The more illiquid any market is, the more likely there will be a dependence on centralized market makers to take the counter-side of the bet. The dependency on market makers introduces inefficiencies in the market from the lack of diverse liquidity.

The use of RMM-01 mitigates this liquidity problem, allowing the growth of a single theta vault to grow as large as the spot market for the asset in question. The RMM-01 trading function explicitly captures theta decay through the swap fee~\cite{primitive} on the underlying spot market, shifting the requirement of a counter-party to option contracts to a need for spot arbitrage alone. This benefits from being a much stronger dependency as spot arbitrage will always be profitable, regardless of price action.

RMM-01 is a CFMM and requires two assets to mint an LPT, with each pool configuration requiring a differing amount of each asset. When the vault switches to an RMM-01 pool, it will need to re-balance the assets in hand before minting the new LPT. When a whole vault consisting of many liquidity providers is looking to re-balance, it can have a significant price impact~\cite{uniswap}. Therefore there is a natural upper bound on the liquidity managed by the vault. Above this liquidity bound, an alternate approach to re-balancing beyond a simple swap is needed.

We suggest intentionally mispricing the new RMM-01 pool, setting the initial ratio of reserves equivalent to the vaults' current liquidity ratio at re-balancing. Thus, there is an incentive for arbitrage to re-align the price~\cite{ross2013arbitrage,angeris2020improved} and, consequently, the allocation. The loss incurred during this re-balance approach scales homogeneously with liquidity. In contrast, the loss incurred during a swap scales at an increased capacity on every market. This approach, however, only appears viable for re-balance before expiry. This restriction is due to the LPT terminally denominating in one of the assets only.
\section{Liquidation-Free Lending Market} \label{sec:liq}

To provide the context of liquidation-free lending, first examine how a lending market defines a liquidation mechanism. Let $V_c\in(0,\infty)$ be the value of the collateral of the borrow position, and let $V_{debt}\in(0,\infty)$ be the value of the borrowed assets. Define the \emph{loan-to-value ratio} as $LTV=\frac{V_{debt}}{V_c}$ and the \emph{collateral ratio} as
\[C=\frac{V_c}{V_{debt}}.\]

Both these terms are called \emph{health factors} and are central to a lending market's structure~\cite{qin2021empirical}. A defined minimum collateral ratio or maximum loan-to-value ratio is required to keep the borrow position open for a particular lending market. If the collateral falls below this ratio, liquidation will occur~\cite{qin2021empirical, perez2021liquidations, gudgeon2020defi}. A funding source is needed to re-balance the borrow position, maintaining the health factor in ranges. Consequently, this eliminates liquidation risk. We can do this through the use of long options. The following approach applies to lending markets structured around the collateral ratio or the LTV-ratio.

\subsection*{Mechanism}

Let $P_x(t)$ be the market price of asset $x$ at time $t$ with respect to some numéraire. Suppose a borrower opens a position consisting of $n$ collateral assets and $m$ debt assets, with prices at the time of borrow $t=0$, $\{P_{R_{x_1}}(0),...,P_{R_{x_n}}(0),P_{R_{y_1}}(0),...,P_{R_{y_m}}(0)\}$ and reserves $\{R_{x_1},...,R_{x_n},R_{y_1},...,R_{y_m}\}$. The health factor initially reads
\[C_0=\frac{P_{R_{x_1}}(0)R_{x_1}+...+P_{R_{x_n}}(0)R_{x_n}}{P_{R_{y_1}}(0)R_{y_1}+...+P_{R_{y_m}}(0)R_{y_m}}
\]

We can use at-the-money calls and put options the hedge out unfavourable price action on the health factor. Suppose $P_{R_{x_i}}(t)<P_{R_{x_i}}(0)$ and $P_{R_{y_j}}(t)>P_{R_{y_j}}(0)$ for $i\in [1,...,n]$ and $j\in\ [1,...,m]$, and $t>0$. Both $C(t)$ and $LTV(t)$ are re-balanced if each the term $P_{R_{x_i}}(t)R_{x_i}$ and $P_{R_{y_j}}(t)R_{y_j}$ for $i\in [1,...,n]$ and $j\in\ [1,...,m]$ are re-balanced. We'll start with the collateral assets,
\[P_{R_{x_i}}(t)R_{x_i}'=P_{R_{x_i}}(0)R_{x_i}
\]
\[R_{x_i}+\Delta_{R_{x_i}}=\frac{P_{R_{x_i}}(0)}{P_{R_{x_i}}(t)}R_{x_i}
\]
\[\Delta_{R_{x_i}}=R_{x_i}(\frac{P_{R_{x_i}}(0)}{P_{R_{x_i}}(t)}-1)
\]
We need to source $\Delta_{R_{x_i}}$ units of the $i$-th collateral asset. We can use put options to source this funding. Note, a put option's value follows
\[V_{put}=\Phi(-d_2)K-\Phi(-d_1)S(t)
\]
where $K=S_0$. On RMM-01, this is denominated in the cash asset of the pool. In terms of the asset $R_{x_i}$, \[V_{put}(P_{R_{x_i}}(t))=\frac{P_{cash}(t)}{P_{R_{x_i}}(t)}(\Phi(-d_2)\frac{P_{R_{x_i}}(0)}{P_{cash}(0)}-\Phi(-d_1)\frac{P_{R_{x_i}}(t)}{P_{cash}(t)})
\]
We can calculate the amount of puts we need, $\alpha_i$, to source $\Delta_{R_{x_i}}$ units of collateral asset $i$,
\[\Delta_{R_{x_i}}=R_{x_i}(\frac{P_{R_{x_i}}(0)}{P_{R_{x_i}}(t)}-1)=\alpha_iV_{put}(P_{R_{x_i}}(t))
\] \\
\[\alpha_i=\frac{x_i(\frac{P_{R_{x_i}}(0)}{P_{R_{x_i}}(t)}-1)}{\Phi(-d_2)\frac{P_{R_{x_i}}(0)}{P_{cash}(0)}-\Phi(-d_1)\frac{P_{R_{x_i}}(t)}{P_{cash}(t)}}\frac{P_{R_{x_i}}(t)}{P_{cash}(t)}
\]
\begin{figure}[H]
    \centering
    \includegraphics[scale=0.5]{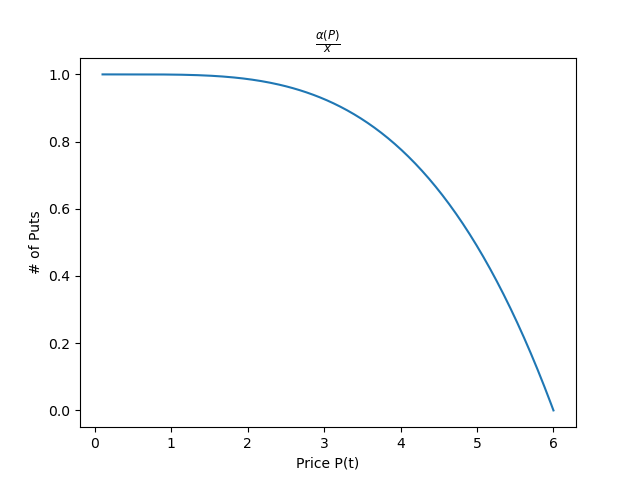}
    \caption{$\frac{\alpha(P)}{x}$ with implied volatility $\sigma=0.85$ and $8$ months to maturity.}
\end{figure}

As shown above in Figure 1, these functions $\alpha_i(P_{R_{x_i}}(t))$ are monotonically decreasing and satisfy the limit
\[\lim_{P_{R_{x_i}}(t)\to0^+}\alpha_i=R_{x_i}
\]
This implies simply holding $R_{x_i}$ amount of at-the-money put options on the $i$-th collateral asset is enough to hedge out any amount of downward price action on the terms $P_{R_{x_i}}(t)R_{x_i}$. Now for the debt assets,
\[P_{R_{y_j}}(t)R_{y_j}'=P_{R_{y_j}}(0)R_{y_j}
\]
\[R_{y_j}-\Delta_{R_{y_j}}=\frac{P_{R_{y_j}}(0)}{P_{R_{y_j}}(t)}R_{y_j}
\]
\[\Delta_{R_{y_j}}=R_{y_j}(1-\frac{P_{R_{y_j}}(0)}{P_{R_{y_j}}(t)})
\]
Again, we need to source $\Delta_{R_{y_j}}$ units of the $j$-th debt asset. Note, a call options value follows
\[V_{call}=\Phi(d_1)S(t)-\Phi(d_2)K
\]
where again $K=S_0$. On RMM-01, this value is denominated in the volatile asset (a.k.a. the debt asset),
\[V_{call}(P(R_{x_i},t))=\frac{P_{cash}(t)}{P_{R_{y_j}}(t)}(\Phi(d_1)\frac{P_{R_{y_j}}(t)}{P_{cash}(t)}-\Phi(d_2)\frac{P_{R_{y_j}}(0)}{P_{cash}(0)})
\]
We can than calculate $\beta_j$, the amount of call options required to source $\Delta_{R_{y_j}}$,
\[\Delta_{R_{y_j}}=R_{y_j}(1-\frac{P_{R_{y_j}}(0)}{P_{R_{y_j}}(t)})=\beta_{j}V_{put}
\]
\[\beta_j=\frac{y_j(1-\frac{P_{y_j}(0)}{P_{y_j}(t)})}{\Phi(d_1)\frac{P_{y_j}(t)}{P_{cash}(t)}-\Phi(d_2)\frac{P_{y_j}(0)}{P_{cash}(0)}}\frac{P_{y_j}(t)}{P_{cash}(t)}
\]

\begin{figure}[H]
    \centering
    \includegraphics[scale=0.5]{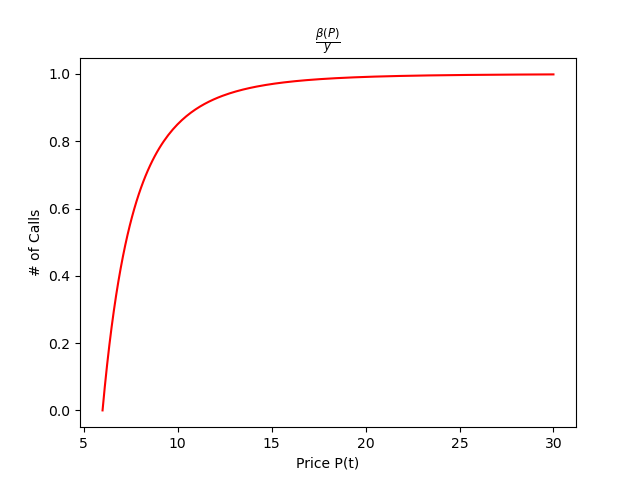}
    \caption{$\frac{\beta(P)}{y}$ with implied volatility $\sigma=0.85$ and $8$ months to maturity.}
\end{figure}

Similar to the case of the collateral assets, as shown above in Figure 2, these functions $\beta_j(P_{y_j}(t))$ are monotonically increasing and satisfy the limit
\[\lim_{P_{y_j}(t)\to\infty}\beta_j=y_j
\]
implying again holding simply $y_j$ at-the-money call options on the $j$-th debt asset is enough to hedge out any amount of positive price action on the terms $P_{R_{y_j}}R_{y_j}$.

Assume the health factor is currently within its limit. Given the reserves of collateral assets and debt assets, we need to maintain precisely at-the-money put options on the collateral assets respectively and at-the-money call options on the debt assets. The resulting hedge prevents the health factor from dipping below its value at initial entry.

Suppose adverse price action occurs, and we need to exercise these options. In the case of RMM-01, this involves repaying the LPT debt and retaining $\{R_{x_1}-\Delta_{R_{x_1}},...,R_{x_n}-\Delta_{R_{x_n}}\}$ of the at-the-money put options and $\{R_{y_1}-\Delta_{R_{y_1}},...,R_{y_m}-\Delta_{R_{y_m}}\}$ of the at-the-money call options. In that case, we will no longer have enough open options to hedge out any amount of adverse price action. One needs to re-insure the position with more opportunities to replace the exercised positions.
\subsection*{Analysis}

The mechanism above assumes that the options in use are at the money, all have the same implied volatility and time to maturity. Options built on an RMM-01~\cite{primitive} pool can never sustain being at the money throughout maturation. Assuming there is sufficient LPT liquidity for lenders, it is reasonable to have options with the same tenor.  As a result of the impracticality of the assumptions in a real-world setting, it is necessary to adjust the above formulation to account for more realistic conditions.

To determine the impact of moneyness on the mechanism's efficiency, we varied the strike price of options in $\alpha_i$ and $\beta_j$. We plotted the resulting surface functions reduced by their respective reserve quantities. As shown below in Figures 3 and 4, strike variance can undoubtedly have a widely variable effect on the number of options required. The further out of the money the options are, the more inefficient the hedging is.

\begin{figure}[H]
    \centering
    \includegraphics[scale=0.5]{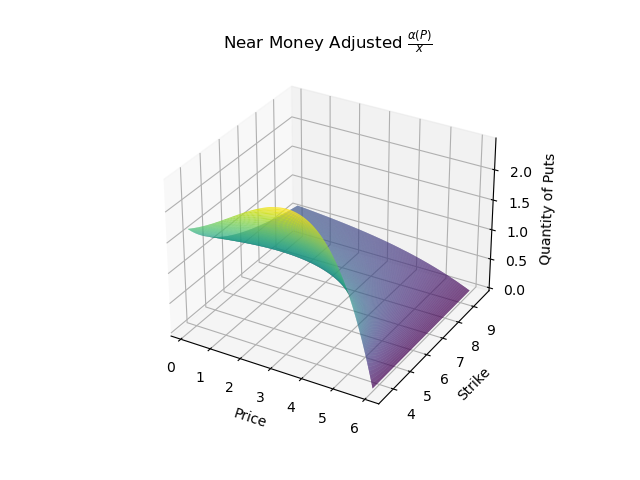}
    \caption{$\frac{\alpha(P)}{R_x}$ Adjusted to near-the-money strike prices, with implied volatility $\sigma=0.85$ and $8$ months to maturity.}
\end{figure}

\begin{figure}[H]
    \centering
    \includegraphics[scale=0.5]{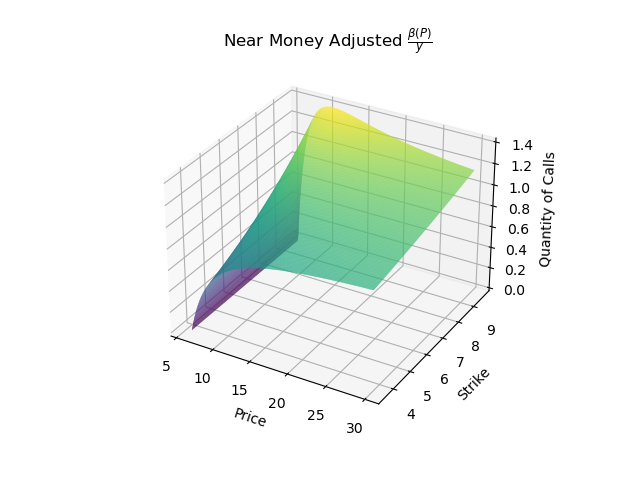}
    \caption{$\frac{\beta(P)}{R_y}$ Adjusted to near-the-money strike prices, with implied volatility $\sigma=0.85$ and $8$ months to maturity.}
\end{figure}

There are a few properties to note. Namely, $\alpha_i$ and $\beta_j$ are no longer monotonic when the options are out of the money. They each exhibit a global maximum at a distinct price point. Consequently, we can no longer look at the values they converge to the ends of their domains. Now we must look at their global maximums to obtain the required options.

If we expect strike prices to be near the money, $\tau$ will be short-dated with a high probability. Hence, the long options will be much less valuable. Being out-of-the-money on the number of options required to hold becomes far more pronounced. This mechanism may be un-viable if the only available long options are significantly out of the money.

The effect of being near the money concludes that the mechanism is still viable for implementation, as the overall efficiency is not impacted too unreasonably for the borrow positions so long as the options aren't deeply out of the money. 
\section{Combined Strategies}\label{sec:combo}

The composable nature of the products allows us to create even further derivatives and payoff structures. Theoretically, we can construct every other payoff function by using strictly cash-or-nothing calls and puts. This construction assumes liquidity at every possible strike price. We describe two basic examples below of products composed of the above instruments.
\subsection*{Straddles}

A straddle is a portfolio composed of an equal quantity of long calls and long puts with the same strike $K$, expiry $T$, and implied volatility $\sigma$~\cite{goldman1979path}. Both long options require shorting the LPT with the corresponding option quantities for the stable or risky asset. To borrow LPTs, we need to provide collateral of either $1-\frac{V_{LPT}}{S(t)}$ risky tokens per LPT or $K-V_{LPT}$ cash tokens per LPT. The asset this collateral is denominated in determines the type of long options achieved. Suppose a borrower opens a position of $m$ call options. Meaning they have contributed $m(1-\frac{V_{LPT}}{S(t)})$ risky tokens to receive a net position of $m(1-\frac{V_{LPT}}{S(t)})$ risky tokens. This borrower's position can be turned into a straddle by simply opening $m$ long put options and holding the $m$ calls.

Thus factoring in the cost of the long puts, the net cost to the borrower for opening a straddle of $m$ calls and $m$ puts is than given by $m(1-\frac{V_{LPT}}{S(t)})$ risky tokens plus $m(K-V_{LPT})$ cash tokens, receiving a net payoff of $m(1-\frac{V_{LPT}}{S(t)})$ risky tokens and $m(K-V_{LPT})$ stable tokens.

Given $x$ risky assets, one can determine the maximum amount of straddles obtainable. 
\begin{equation}
    x=m(1-\frac{V_{LPT}}{S(t)}+\frac{K-V_{LPT}}{S(t)})
\end{equation}
\begin{equation}
    m=\frac{x}{1-\frac{V_{LPT}}{S(t)}+\frac{K-V_{LPT}}{S(t)}}
\end{equation}

Symmetry in directional exposure depends on an equal quantity of calls and puts\cite{carr1999static}. The attraction of this symmetry is long volatility without leaving a directional bias in the position. If we expect volatility to rise but still have a directional bias, tune these quantities to fit the bias. For example, suppose we hold a bullish directional bias. It is advantageous to open more calls than puts to skew the payoff more directionally bullish.

On RMM-01, since both long options only depend on the LPTs, we can effectively take out any symmetry straddle on whatever $K$, $T$, $\sigma$ configuration pool with LPTs lent.

\subsection*{Futures}

A future is a derivative contract that forces two parties to transact an asset with each other at a pre-set date and price~\cite{tashjian1995optimal}. The nature of parity in options with expiration-dependent payoffs allows the recreation of futures positions using options. While futures are achievable on RMM-01, it is prohibitively more expensive than futures on centralized exchanges due to that lack of available leverage. It is still a worthwhile mechanism to define on RMM-01, as it provides further insight into the interaction of RMM-01 based derivatives and structured products.

A long future position payoff resembles holding spot at expiry. To achieve this, we need to open a covered call and a call option with the same $K$, $T$, $\sigma$ configuration. To open a covered call position, we need to provide liquidity to the RMM-01 of our choice. The cost of which is
\begin{equation}
    x=\Phi(-\frac{\ln(S/K)}{\sigma\sqrt{\tau}}-\frac{1}{2}\sigma\sqrt{\tau})
\end{equation}
risky tokens and $y(x)$ stable tokens, where $S(t)$ is the current spot price of the pool. To open a call options, the cost is $1-\frac{V_{LPT}}{S(t)}$, implying a debt of $1$ covered call. The net cost of opening a long future position is
\begin{equation}
1
\end{equation}
which is where the inefficiency arises as this implies full collateralization for any future position. Omitting the use of RMM-01, holding $1$ unit of the risky until the expiration time would lead to the same payoff function. Moreover, this is the idealized case with no trade slippage. Factoring in this additional friction, we conclude that there are significant savings by staying in the risky asset from the start. 

We can extend the above construction to include short future positions by shorting the long future position. We construct a short future position from a covered put and a long put with the same $K$, $T$, $\sigma$ configuration. On RMM-01, constructing a short future can be approached by shorting a call option for the stable asset to achieve a covered put and then opening a long put position.

This approach requires an additional lending market on the long options, which acts as another layer of friction and liquidity fragmentation. Thus, significantly increasing the cost beyond the extent of the long future case. Obtaining a short future position on RMM-01 is unreasonable due to additional friction.  

\section{Conclusion}

We have provided a theoretical demonstration of utilizing RMM-01 to construct a diverse array of financial primitives. Vanilla and binary options mechanisms implement financial primitives on-chain familiar to traditional financial marketplaces. Mechanisms for liquidation-free lending allow for a better user experience by mitigating liquidation risk. The theta vault framework provides a strategy to maintain vault performance at scale. 
\subsection*{Future Work}

This paper only briefly constructs each mechanism and does not analyze their behaviors in real-world market conditions. The capital efficiency of theta vaults and liquidation-free lending on RMM-01 needs careful evaluation. This work lays the foundation for future work in permissionless derivative design. Further inquiry into MEV-aware application design, mechanism architecture, and transparent structured products is necessary.
 

\bibliography{citations.bib}
\end{document}